\documentclass[a4paper,10pt]{article}

\usepackage{fouriernc}
\usepackage{amsthm}
\usepackage{amsmath}
\usepackage{amssymb}

\newtheorem{theorem}{Theorem}[section]

\newtheorem{lemma}[theorem]{Lemma}

\newtheorem{remark}[theorem]{Remark}

\newtheorem{example}[theorem]{Example}
\newtheorem{assumption}[theorem]{Assumption}

\begin{document}

\title{On utility maximization without passing by the dual problem\thanks{The
author thanks Freddy Delbaen and Keita Owari for discussions about Section \ref{or} and an anonymous
referee for very useful comments that led to substantial improvements. Special thanks go to Ngoc Huy Chau for 
discussions which helped discovering and removing an error. The support received from the ``Lend\"ulet'' grant LP 2015-6 of the
Hungarian Academy of Sciences and from the NKFIH (National Research, Development and Innovation Office, Hungary) 
grant KH 126505 is gratefully acknowledged.}}

\author{Mikl\'os R\'asonyi}

\date{\today}

\maketitle

\begin{abstract}
We treat utility maximization from terminal wealth 
for an agent with utility function $U:\mathbb{R}\to\mathbb{R}$ who dynamically invests
in a continuous-time financial market and receives a 
possibly unbounded random endowment. We prove the existence of an optimal investment
without introducing the associated dual problem. We rely on a recent result of Orlicz space theory, due to Delbaen
and Owari which leads to a simple and transparent proof.

Our results apply to non-smooth utilities and even strict concavity can
be relaxed. We can handle certain random endowments with non-hedgeable risks, complementing earlier papers.
Constraints on the terminal wealth can also be incorporated.

As examples, we treat frictionless markets with finitely many assets and large financial markets.  
\end{abstract}

\section{Prologue}

Utility maximization from terminal wealth with a random endowment is known to be delicate 
as complications for the dual problem arise. This was first noticed in \cite{csw}. 
Here we propose a method to prove the existence of
maximizers working on the primal problem only, for utility functions $U$ that are finite on the whole real line. 
This method allows the treatment of random endowments without tackling the dual problem. Constraints
on the terminal wealth can also be easily incorporated. 
The proofs are transparent and rather straightforward.
We utilize a Koml\'os-type compactness result of \cite{delbaen-owari}, see
Lemma \ref{do} below. 

When $U$ is defined on $(0,\infty)$, a direct approach to the primal problem of utility maximization
is well-known from \cite{wa}, and it has already been exploited in markets with constraints (see \cite{lz}) 
or with frictions (see \cite{pao,3c}). For $U$ with domain $\mathbb{R}$ our method seems the first to avoid
solving the dual problem. The conjugate function of $U$ does appear also in our approach, we use
Fenchel's inequality and some Orlicz space theory but the dual problem does not even need to be
defined. 

After reviewing facts of Orlicz space theory in Section \ref{or},
we formulate Theorems \ref{main1}, \ref{main1.5}, \ref{main2}, \ref{const} 
and \ref{constt} in Section \ref{ut} 
in a general setting, without reference to specific types of market models. Then we demonstrate the power of our
method by considering frictionless markets with finitely many assets (Section \ref{fric})
and large financial markets (Section \ref{lar}).

Besides displaying a new, simple method, our paper makes several contributions improving on the
existing literature. We are listing them now.

\begin{itemize}

\item  In frictionless markets, Theorems \ref{mai1} and \ref{mai2} allow unbounded, possibly non-hedgeable random endowments, see Remark \ref{burt} and Example \ref{unbi} for details. The asset prices need not be locally bounded.
We do not require smoothness of $U$ and strict concavity is not imposed either.  
In particular, we provide minimizers for loss functionals, see Example \ref{lil}. Constraints of a very general
type on the terminal portfolio value are admitted in Theorem \ref{conss}. 

\item In the theory of large financial markets, our approach is the first to tackle utility maximization for $U$
finite on $\mathbb{R}$ (the case of $U$ defined on $(0,\infty)$ was first considered in \cite{paolo}; 
subsequently \cite{mostovyi} treated random endowments in the same setting), see Section \ref{lar}.

\end{itemize}

\section{About Orlicz spaces}\label{or}

Orlicz spaces as the appropriate framework for utility maximization have already been
advocated in \cite{sara_orlicz,sb,sara}. These spaces play a crucial role in our approach, too.

We write $x_+$ (resp. $x_-$) to denote the positive (resp. negative) part of some $x\in\mathbb{R}$. 
Fix a probability space $(\Omega,\mathcal{F},P)$. We identify random variables differing on a $P$-zero set only.
We denote by $L^0$ the set of all $\mathbb{R}$-valued random variables. The family of non-negative elements 
in $L^0$ is denoted by $L^0_+$. The symbol $EX$ denotes
the expectation of $X\in L^0$ whenever this is well-defined (i.e. either $EX_+$ or $EX_-$ is finite). 
If $Q$ is another probability on $\mathcal{F}$ then
the $Q$-expectation of $X$ is denoted by $E_Q X$. Let $L^1(Q)$ denote the usual Banach space
of $Q$-integrable random variables on $(\Omega,\mathcal{F},Q)$ for some probability $Q$. When $P=Q$
we simply write $L^1$. A reference work for
the results mentioned in the discussion below is \cite{rr}.

In this paper, we call $\Phi:\mathbb{R}_+\to\mathbb{R}_+$ a \emph{Young function} if it is convex with $\Phi(0)=0$ and
$\lim_{x\to\infty}\Phi(x)/x=\infty$. The set 
\[
L^{\Phi}:=\{X\in L^0:\ E\Phi(\gamma |X|)<\infty\mbox{ for some }\gamma>0\}
\]
becomes a Banach space (called the Orlicz space corresponding to $\Phi$) with the norm 
\[
\Vert X\Vert_{\Phi}:=\inf\{\gamma>0:X\in \gamma B_{\Phi}\},
\]
where $B_{\Phi}:=\{X\in L^0:\ E\Phi(|X|)\leq 1\}$. Define the conjugate function $\Phi^*(y):=\sup_{x\geq 0}[xy-\Phi(x)]$,
$y\in\mathbb{R}_+$.
This is also a Young function and we have 
\[
(\Phi^{*})^*=\Phi.
\]

We say that $\Phi$ is \emph{of class} $\Delta_2$ if 
\[
\limsup_{x\to\infty}\frac{\Phi(2x)}{\Phi(x)}<\infty.
\]
In this case also 
\[
h_{\Phi}:=\sup_{x\geq 1}\frac{\Phi(2x)}{\Phi(x)}<\infty.
\]

\begin{remark}\label{amu} {\rm Let $\Phi$ be a Young function. Then $E\Phi(|X|)<\infty$
implies $\Vert X\Vert_{\Phi}<\infty$ and the two conditions are equivalent when $\Phi$
is of class $\Delta_2$. These well-known observations play an important role in our arguments so
we provide their proofs here, for convenience.
  
Assume first that $E\Phi(|X|)<\infty$.
Convexity of $\Phi$ and $\Phi(0)=0$ imply $\Phi(x/m)\leq \Phi(x)/m$ for all $m\geq 1$ and $x\geq 0$. Take $X$ with $E\Phi(|X|)=:M<\infty$. Then 
$E\Phi(|X|/(M+1))< 1$ hence, by definition, 
$\Vert X\Vert_{\Phi}\leq M+1<\infty$. 

Looking at the converse direction: if $\Phi$ is a Young function of class
$\Delta_2$ then for any $X\in L^0$, $\Vert X\Vert_{\Phi}<2^k$ implies 
\begin{equation}\label{ree}
E\Phi(|X|)\leq h_{\Phi}^k E\Phi(|X|/2^k) +\Phi(2^k)\leq h_{\Phi}^k+\Phi(2^k),
\end{equation} 
for all integers $k\geq 1$. It follows that if $\Vert X\Vert_{\Phi}=:M'<\infty$
then, for $k$ large enough, $2^k>M'$ hence also $E\Phi(|X|)<\infty$, by \eqref{ree}.}
\end{remark}

Let us recall a compactness result of \cite{delbaen-owari} which is crucial for the developments 
of the present paper.

\begin{lemma}\label{do} Let $\Phi$ be a Young function of class $\Delta_2$ 
and let $\xi_n$, $n\geq 1$ be a norm-bounded sequence in $L^{\Phi^*}$.
Then there are convex weights $\alpha_j^n\geq 0$, $n\leq j\leq M(n)$,
$\sum_{j=n}^{M(n)} \alpha_j^n=1$ such that 
\[
\xi_n':=\sum_{j=n}^{M(n)}\alpha_j^n \xi_j
\]
converges almost surely to some $\xi\in L^{\Phi^*}$ and $\sup_n |\xi_n'|$ is in $L^{\Phi^*}$.
\end{lemma} 
\begin{proof}
This is Corollary 3.10 of \cite{delbaen-owari}.
\end{proof}

\section{A general framework}\label{ut}

In this section no particular market model is fixed. Instead, an abstract framework is
presented where portfolios are represented by their wealth processes which are assumed
to be supermartingales under a certain set of reference probability measures. We deduce
the existence of optimal portfolios in such a setting.

Let $T>0$ be a fixed finite time horizon and let $(\Omega,\mathcal{F},(\mathcal{F}_t)_{t\in [0,T]},P)$ 
be a stochastic basis satisfying the usual hypotheses. 

Our requirements on the utility function are summarized in the following assumptions.

\begin{assumption}\label{u}
The function $U:\mathbb{R}\to\mathbb{R}$ is nondecreasing and concave, $U(0)=0$. 
Define the convex conjugate of $U$ by 
\[
V(y):=\sup_{x\in\mathbb{R}}[U(x)-xy].
\]
We stipulate
\begin{eqnarray}
\lim_{x\to-\infty}\frac{U(x)}{x} &=&\infty,\label{levente}\\
\limsup_{y\to\infty}\frac{V(2y)}{V(y)} &<& \infty.\label{hetente}
\end{eqnarray}
\end{assumption}

\begin{assumption}\label{moder} Let $U:\mathbb{R}\to\mathbb{R}$ be such that
\begin{equation}\label{berente}
\limsup_{x\to -\infty}\frac{U(2x)}{U(x)} < \infty.
\end{equation}
\end{assumption}

In simple terms, a risk-averse investor is considered who prefers more to less. In many of the related studies,
$U$ is also assumed continuously differentiable and strictly concave. For
purposes of e.g. loss minimization, however, strict concavity of $U$ would be
too much to require. 

We remark that \eqref{levente} implies $V(y)>0$ for $y$ large enough
hence \eqref{hetente} makes sense. We also point out that \eqref{hetente} is implied by the standard ``reasonable asymptotic elasticity''
condition, see e.g. Corollary 4.2 of \cite{walter}, hence \eqref{hetente} is rather mild
a hypothesis. However, condition \eqref{berente} is, admittedly,
quite restrictive since it excludes e.g. the exponential utility. 

Define 
\begin{equation}\label{mathp}
\mathcal{P}_V:=\{Q\ll P:\ EV(dQ/dP)<\infty\} 
\end{equation}
and let $\mathcal{M}^a_V\subset\mathcal{P}_V$ be a fixed set of ``reference probabilities''. 
Denote 
$$
\mathcal{M}^e_V:=\{Q\in\mathcal{M}^a_V:\, Q\sim P\}.
$$
Notice that no convexity or closedness assumption is required about the ``set of reference
probabilities'' $\mathcal{M}^a_V$.

We introduce
\begin{align}\nonumber
\mathcal{S} := \{Y_t,\ t\in [0,T] &: \ Y\mbox{ is a c\`adl\`ag }\\
&R\mbox{-supermartingale,}\mbox{ for all }R\in\mathcal{M}^a_V,\ Y_0=0\}.\label{chante}
\end{align}
Clearly, $\mathcal{S}\neq\emptyset$ since the identically zero supermartingale is therein. 
Also, $\mathcal{S}$ is convex. 

We now stipulate our conditions on the random endowment $\mathcal{E}$ that the investor receives.
\begin{assumption}\label{m0} There exists $Q\in \mathcal{M}^e_V$. 
$\mathcal{E}$ is $\mathcal{F}_T$-measurable and, for all $R\in\mathcal{M}^a_V$, 
$E_R|\mathcal{E}|<\infty$.
\end{assumption}

\begin{remark}\label{ludlo} {\rm We provide a simple sufficient condition for Assumption \ref{m0} under
the objective probability $P$. Let Assumption \ref{u} be in force and consider the conditions
\begin{equation}\label{rom2}
EU(-\mathcal{E}_+)>-\infty,\ EU(-\mathcal{E}_-)>-\infty.
\end{equation}
They can be interpreted as ``gains or losses from the random endowment
should not be too large'' as measured by the tail of $U$ at $-\infty$. Notice that, by the 
Fenchel inequality, 
for \emph{any} $R\in \mathcal{M}^a_V$,
\[
E_R\mathcal{E}_{\pm}\leq EV(dR/dP)-EU(-\mathcal{E}_{\pm})<\infty,
\]
by \eqref{rom2}. We conclude that Assumption \ref{m0} holds for every 
$\mathcal{F}_T$-measurable $\mathcal{E}$ satisfying
\eqref{rom2} provided that $\mathcal{M}_V^e\neq\emptyset$.}
\end{remark}

We fix a non-empty convex subset $\mathcal{A}\subset\mathcal{S}$, its elements will correspond to
``admissible'' portfolios, depending on the context.
We imagine that, for each $Y\in\mathcal{A}$, $Y_T$ represents the value at $T$ of an available
investment opportunity (e.g. the terminal wealth of a dynamically rebalanced portfolio in the
given market model). In general, $\{ Y_T:\, Y\in\mathcal{A}\}$ is not
closed in any suitable sense so it is desirable to carry out
utility maximization over a larger class of processes. Such a class is defined now. 
\begin{eqnarray}\nonumber
\mathcal{A}_U &:=& \{Y\in\mathcal{S}:\mbox{ there is }Y^n\in\mathcal{A}
\mbox{ with }U(Y^n_T+\mathcal{E})\in L^1,\ n\geq 1\\
& &\mbox{ and }U(Y^n_T+\mathcal{E})\to U(Y_T+\mathcal{E}),\, n\to\infty,\mbox{ in }L^1\}.\nonumber 
\end{eqnarray}

\begin{remark}{\rm Choosing the domain of optimization is a subtle issue when $U:\mathbb{R}\to\mathbb{R}$.
The above definition follows the choice of \cite{walter}. In that paper (and in many subsequent studies),
$\mathcal{A}$ is the set of portfolio value processes that are bounded from below
(these all lie in $\mathcal{S}$) and the domain of optimization is its ``closure'' $\mathcal{A}_U$.}
\end{remark}



\begin{theorem}\label{main1} Let Assumptions \ref{u}, \ref{moder} and \ref{m0} be in force and let
$U$ be bounded above. Then there exists
$Y^{\dagger}\in\mathcal{A}_U$ such that
\[
EU(Y^{\dagger}_T+\mathcal{E})=\sup_{Y\in\mathcal{A}_U}EU(Y_T+\mathcal{E}),
\]
provided that $\mathcal{A}_U\neq \emptyset$.
\end{theorem}

\begin{remark}\label{rem} {\rm 
A sufficient condition for $\mathcal{A}_U\neq\emptyset$ is
\begin{equation}\label{kot}
EU(\mathcal{E})>-\infty,\ 0\in\mathcal{A},
\end{equation} 
since $0\in\mathcal{A}_U$ in that case. Actually, 
under Assumption \ref{moder}, it is not difficult to show that \eqref{kot} implies $EU(\mathcal{E}+z)>-\infty$
for all $z\leq 0$ as well. So, under \eqref{kot}, $Y\in \mathcal{A}_U$ whenever $Y\in\mathcal{A}$ and $Y_T$ is bounded.}
\end{remark}

\begin{proof}[Proof of Theorem \ref{main1}.] Let $\beta$ denote the left derivative of
$U$ at $0$. Define $\Phi^*(x):=-U(-x)$, $x\geq 0$.
Its conjugate equals \[
\Phi(y):=\Phi^{**}(y)=\left\{
\begin{aligned}
0,&\quad\mbox{ if }0\leq y\leq\beta,\\
V(y)-V(\beta),&\quad\mbox{ if }y>\beta,
\end{aligned}
\right. 
\]
see e.g. \cite{sara_orlicz}. $\Phi^*$ is a Young function  
by \eqref{levente} hence $\Phi$ is also a Young function which is of class $\Delta_2$ by \eqref{hetente}.

Let $Y^n\in\mathcal{A}_U$, $n\in\mathbb{N}$ be such that 
\begin{equation}\label{maxmi}
EU(Y^n_T+\mathcal{E})\to\sup_{Y\in\mathcal{A}_U}EU(Y_T+\mathcal{E}),\ n\to\infty,
\end{equation}
where the latter supremum is $>-\infty$ by $\mathcal{A}_U\neq \emptyset$.
By definition of $\mathcal{A}_U$ we may and will suppose that $Y^n\in\mathcal{A}$, $n\geq 1$.

Fix $Q$ as in Assumption \ref{m0}.
By the Fenchel inequality,
\begin{equation}\label{bross}
E_Q[Y^n_T+\mathcal{E}]_-\leq E\Phi(dQ/dP) -EU(-[Y^n_T+\mathcal{E}]_-) 
\end{equation}
and, by $U(0)=0$, we have $-U(-[Y^n_T+\mathcal{E}]_-)=[U(Y^n_T+\mathcal{E})]_-$.

Let $C\geq 0$ denote an upper bound for $U$. We must have
\begin{equation}\label{mam1}
\sup_n E[U(Y^n_T+\mathcal{E})]_-<\infty
\end{equation}
otherwise 
\[
\inf_n\, EU(Y^n_T+\mathcal{E})\leq C-\sup_n\, E[U(Y^n_T+\mathcal{E})]_-=-\infty
\]
would hold which clearly contradicts \eqref{maxmi}. 

We may and will suppose that $T\in\mathbb{Q}$.
As $Y^n$ is a supermartingale under $Q$, we have that $[Y^n_t]_-$, $t\in [0,T]$ is
a $Q$-submartingale hence we get, for all $t\in\mathbb{Q}\cap [0,T]$,
\begin{eqnarray*}
\sup_n E_Q\vert Y^n_t\vert &\leq& 2\sup_n E_Q[Y^n_t]_- \leq\\
 2\sup_n E_Q[Y^n_T]_- &\leq& 2\sup_n E_Q[Y^n_T+\mathcal{E}]_- + 2E_Q\mathcal{E}_+<\infty,
\end{eqnarray*}
by \eqref{bross}, \eqref{mam1} and Assumption \ref{m0}
so the theorem of Koml\'os and a diagonal argument imply the existence of a subsequence (which we continue to
denote by $n$) such that \[
\tilde{Y}^n:=\frac{1}{n}\sum_{j=1}^n Y^j\in \mathcal{A},\quad n\geq 1,
\]
satisfy $\tilde{Y}^n_t\to \tilde{Y}^{\dagger}_t$ $Q$-almost surely (and hence also $P$-almost surely) for $t\in [0,T]\cap\mathbb{Q}$ where
$\tilde{Y}^{\dagger}_t$, $t\in [0,T]\cap\mathbb{Q}$ is a (finite-valued) process. 

As $U$ is concave, we have 
\[
EU(\tilde{Y}^n_T+\mathcal{E})\to\sup_{Y\in\mathcal{A}_U}EU(Y^n_T+\mathcal{E}),\ n\to\infty,
\]
as well as
\begin{equation}\label{mam}
\sup_n E\Phi^*([\tilde{Y}^n_T+\mathcal{E}]_-)=\sup_n E[U(\tilde{Y}^n_T+\mathcal{E})]_-
\leq \sup_n E[U(Y^n_T+\mathcal{E})]_-<\infty.
\end{equation}


Set $\xi_n:=[\tilde{Y}^n_T+\mathcal{E}]_-$.
Note that
\[
\sup_n E\Vert \xi_n\Vert_{\Phi^*}<\infty,
\]
by \eqref{mam} and Remark \ref{amu}.
Applying Lemma \ref{do}, we get convex weights $\alpha_j^n\geq 0$, $n\leq j\leq M(n)$,
$\sum_{j=n}^{M(n)} \alpha_j^n=1$ such that 
\[
Z^n:=\sum_{j=n}^{M(n)}\alpha_j^n \xi_n,\ n\geq 1
\]
satisfy
\begin{equation}\label{zazen}
L:=\Vert \sup_n Z^n\Vert_{\Phi^*}+1<\infty.
\end{equation}
Now define 
\[
\overline{Y}^n:=\sum_{j=n}^{M(n)}\alpha_j^n \tilde{Y}^n\in\mathcal{A},\ n\geq 1,
\]
and set 
$w_T:=\sup_n\left(\overline{Y}^n_T+\mathcal{E}\right)_-.$ 

We claim that $w_T$ is $R$-integrable for all $R\in \mathcal{M}^a_V$.
Indeed,
using convexity of the mapping $x\to x_-$,
\begin{eqnarray}\label{nedelni}
w_T \leq
\sup_{n}\sum_{j=n}^{M(n)}\alpha_j^n (\tilde{Y}^j_T +\mathcal{E})_- \leq
\sup_n Z^n.
\end{eqnarray}

By the Fenchel inequality and \eqref{zazen},
\begin{eqnarray*}
E_R\sup_n Z^n &\leq& LE_R\left[\frac{\sup_n Z^n}{L}\right]\leq 
LE\Phi(dR/dP)+ LE\Phi^*\left(\frac{\sup_n Z^n}{L}\right)\leq \\
LE\Phi(dR/dP)+ L &<& \infty,
\end{eqnarray*}
which shows the claim. 

Now take an arbitrary $R\in\mathcal{M}^a_V$. 
Define the $R$-martingale $\varepsilon_t^R:=E_R[\mathcal{E}\vert\mathcal{F}_t]$,
$t\in [0,T]$ and note that 
$\varepsilon^R_T=\mathcal{E}$ by
Assumption \ref{m0}. 
Define also the $R$-martingale 
$$
w_t^R:=E_{R}\left[w_T\vert
\mathcal{F}_t\right],\ t\in [0,T],
$$ 
(we take a c\`adl\`ag version for both $w^R$ and $\varepsilon^R$).

Since, clearly, $\overline{Y}_t\to \tilde{Y}^{\dagger}_t$ for $t\in [0,T]\cap\mathbb{Q}$ and, by the $R$-submartingale property of $[\overline{Y}_t+\varepsilon^R_t]_-$, $t\in [0,T]$,
\[
\sup_n\, E_R[\overline{Y}_t+\varepsilon^R_t]_-\leq \sup_n\, E_R\left[ [\overline{Y}_T+\varepsilon_T]_-\vert\mathcal{F}_t\right]\leq w^R_t,
\]
we get that $\tilde{Y}^{\dagger}_t+\varepsilon_t^R$, $t\in [0,T]\cap\mathbb{Q}$ is an $R$-supermartingale
for each $R\in\mathcal{M}^a_V$ and so is $\tilde{Y}^{\dagger}_t$, 
$t\in [0,T]\cap\mathbb{Q}$, in particular, this holds for $R=Q\in \mathcal{M}^e_V$. Hence also 
\[
Y^{\ddagger}_t:=\lim_{s\in\mathbb{Q}\cap [0,T],s\downarrow t} \tilde{Y}^{\dagger}_s,\ t\in [0,T),\ Y^{\ddagger}_T:=\tilde{Y}^{\dagger}_T,
\]
(where the limit exists $Q\sim P$-almost surely), is a c\`adl\`ag $R$-supermartingale, using $\tilde{Y}_t^{\dagger}\geq -w^R_t-\varepsilon^R_t$,
$t\in [0,T]$. As this argument works for every $R\in\mathcal{M}_V^a$, it follows that 
\begin{equation}\label{eddig}
Y^{\ddagger}\in \mathcal{S}.
\end{equation}

Note that, up to this point, we have not used Assumption \ref{moder} yet.
The function $\Phi^*$ is of class $\Delta_2$ by \eqref{berente}. Hence
from \eqref{zazen} and Remark \ref{amu}, 
\begin{equation}\label{mable}
E\Phi^*(\sup_n Z^n)<\infty
\end{equation}
follows. Noting \eqref{mable}, \eqref{nedelni} 
and the fact that $U$ is bounded from above, dominated
convergence implies 
\begin{equation}\label{doc}
U(\overline{Y}^n_T+\mathcal{E})\to U({Y}^{\ddagger}_T+\mathcal{E})\mbox{ in }L^1,\ n\to\infty,
\end{equation} 
and, by the construction of the sequence $\overline{Y}^n$, we get that
\[
EU(Y^{\ddagger}_T+\mathcal{E})=\sup_{Y\in\mathcal{A}_U}EU(Y^n_T+\mathcal{E}).
\]
As $Y^{\ddagger}\in\mathcal{A}_U$ holds by \eqref{doc}, we can set $Y^{\dagger}:=Y^{\ddagger}$. 
\end{proof}

As we have pointed out, Assumption \ref{moder} is restrictive and it would be desirable to
drop it. This is possible if we modify our assumptions on the domain of optimization.
A sequence $Y^n\in\mathcal{S}$, $n\in\mathbb{N}$ is called \emph{Fatou-convergent}, 
if $Y^n_T\to Z$, $n\to\infty$ a.s. for some random variable $Z$ and for every $R\in\mathcal{M}_V^a$ there is an $R$-martingale $w^R$ with
$\inf_{n}Y^n_t\geq w^R_t$ a.s., for all $t\in [0,T]$.
A class $\mathcal{I}\subset\mathcal{S}$ is \emph{Fatou-closed} if, for every Fatou-convergent
sequence $Y^n\in\mathcal{I}$, $n\in\mathbb{N}$, there exists $Y\in\mathcal{I}$ with the property $Y_T\geq Z$ a.s. 

\begin{theorem}\label{main1.5}
Let Assumptions \ref{u} and \ref{m0} be in force and let $\emptyset\neq\mathcal{I}\subset\mathcal{S}$
be convex and Fatou-closed. Then there is $Y^{\dagger}\in\mathcal{I}$ such that
\[
EU(Y^{\dagger}_T+\mathcal{E})=\sup_{Y\in\mathcal{I}}EU(Y_T+\mathcal{E}).
\]
\end{theorem}
\begin{proof}
If the supremum is $-\infty$ then there is nothing to prove. Otherwise we follow the steps of the proof
of Theorem \ref{main1} up to \eqref{eddig} but with $\mathcal{A},\mathcal{A}_U$ both replaced by $\mathcal{I}$.
Fatou-closedness of $\mathcal{I}$ implies that there is $Y^{\dagger}\in\mathcal{I}$ with 
$Y^{\dagger}_T\geq Y^{\ddagger}_T$.
By the construction of $\overline{Y}^n$ and by Fatou's lemma,
\[
EU(Y^{\dagger}_T+\mathcal{E})\geq EU(Y^{\ddagger}_T+\mathcal{E})\geq\sup_{Y\in\mathcal{I}}EU(Y_T+\mathcal{E}),
\]
but there must be equalities here since $Y^{\dagger}\in\mathcal{I}$.
\end{proof}

\begin{remark}{\rm The Fatou-closure property of the domain of optimization $\mathcal{I}$ is familiar
from the arbitrage theory of frictionless markets. However, the notion we use is different
from that of e.g. \cite{ds} and it is better adapted to our purposes. 

The definition of $\mathcal{A}_U$ stressed the possibility of 
approximating each element in the domain of optimization by value processes of ``admissible'' strategies (i.e. by strategies from the class $\mathcal{A}$). 
This is a crucial feature in large markets, see Section \ref{lar}. 
The domain of optimization $\mathcal{I}$ can be thought of as being possibly ``larger'', requiring only the supermartingale property for each value process. Considering domains like $\mathcal{I}$ follows the stream of literature 
represented by e.g. \cite{bifri} and \cite{oz}.}
\end{remark}

A weakness of Theorems \ref{main1}, \ref{main1.5} is that $U$ was assumed to be bounded from above. 
One can relax this condition at
the price of requiring more about $\mathcal{M}^e_V$.

\begin{assumption}\label{m2} Let 
\begin{equation}\label{maki}
U(x)\leq D[x^{\alpha}+1],\quad x\geq 0,
\end{equation} 
with some $0\leq \alpha<1$ and $D>0$. $\mathcal{E}$ is $\mathcal{F}_T$-measurable and $E_R|\mathcal{E}|<\infty$ for each $R\in\mathcal{M}_V^a$.
We stipulate the existence of $Q\in\mathcal{M}^e_V$ 
such that $E(dP/dQ)^{r}<\infty$ for some $r>\alpha/(1-\alpha)$. 
\end{assumption}

\begin{remark}
{\rm We explain the meaning of this assumption on a simple example of a utility function. Let $0<\alpha<1$ and
$\beta>1$ and set
\[
U(x):=\frac{1}{\alpha}[(1+x)^{\alpha}-1]\mbox{ for }x\geq 0,\quad U(x):=-\frac{1}{\beta}[(1-x)^{\beta}-1]\mbox{ for }x<0.
\]
In this case a direct calculation shows that $Q\in\mathcal{M}^e_V$ implies 
\[
E(dP/dQ)^{\alpha/(1-\alpha)}<\infty,
\] 
but integrability with a higher power $r>\alpha/(1-\alpha)$ does not necessarily hold. What we require in 
Assumption \ref{m2} is thus ``slightly more integrability of $dP/dQ$'' than what is implied by the standard assumption 
on the existence of $Q\in\mathcal{M}^e_V$. It would
be nice to drop this latter condition but we do not know how to achieve this.

We remark that \eqref{maki} is slightly weaker than the standard
condition of ``reasonable asymptotic elasticity'',
see \cite{walter} and Lemma 6.5 of \cite{doks}.}
\end{remark}

\begin{theorem}\label{main2} Let Assumptions \ref{u}, \ref{moder} and \ref{m2} be in force and let $\mathcal{A}_U\neq\emptyset$. Then there exists
$Y^{\dagger}\in\mathcal{A}_U$ such that
\[
EU(Y^{\dagger}_T+\mathcal{E})=\sup_{Y\in\mathcal{A}_U}EU(Y_T+\mathcal{E}).
\]
\end{theorem}


\begin{proof} We only point out what needs to be modified with respect to the proof
of Theorem \ref{main1}. We borrow ideas from \cite{rasonyi}. Take $Q$ as in Assumption \ref{m2}. Recall that $U(0)=0$.

Let $1>\theta>\alpha$ be such that $\theta/(1-\theta)=r$. Let $K\geq 0$.
For any random variable $X$ with $E_Q X\leq K$ we can estimate, using H\"older's and Fenchel's inequalities as well
as the elementary $(x+y)^{\theta}\leq x^{\theta}+y^{\theta}$, $x^{\alpha}\leq x^{\theta}+1$, $x,y\geq 0$, 
\begin{eqnarray}
EU(X_+) &\leq& D[EX_+^{\theta} +2]\leq D[C_1(E_Q X_+)^{\theta}+2] \leq\label{ulm}\\
D[C_1(E_Q X_-+K)^{\theta}+2] &\leq& DC_1[(E_Q X_-)^{\theta}+K^{\theta}]+2D \leq\nonumber\\
& & DC_1 [(E\Phi(dQ/dP)-EU(-X_-))^{\theta}+K^{\theta}]+2D\nonumber, 
\end{eqnarray}
where $C_1:=(E_Q[dP/dQ]^{1/(1-\theta)})^{1-\theta}=(E[dP/dQ]^{\theta/(1-\theta)})^{1-\theta}<\infty$.

Applying \eqref{ulm} to $X:=Y^n_T+\mathcal{E}$ with $K:=E_Q\mathcal{E}_+$, it follows that if we had  
$E_Q[Y^n_T+\mathcal{E}]_-\to\infty$ along a subsequence then we would also have
\begin{equation}\label{tomor}
EU(Y^n_T+\mathcal{E})=EU([Y^n_T+\mathcal{E}]_+)-(-EU(-[Y^n_T+\mathcal{E}]_-))\to -\infty 
\end{equation}
along the same subsequence since $\theta<1$. This contradicts the choice of $Y^n$ so necessarily
\begin{equation*}
\sup_n E_Q[Y^n_T+\mathcal{E}]_-<\infty
\end{equation*}
and then also
\begin{equation}\label{macci}
\sup_n E_Q |Y^n_T+\mathcal{E}|<\infty,
\end{equation} 
since $Y^n$ is a $Q$-supermartingale and Assumption \ref{m2} holds. From \eqref{ulm} it follows that
\[
\sup_n EU([Y^n_T+\mathcal{E}]_+)<\infty. 
\]
The latter observation implies $\sup_n E[U(Y^n_T+\mathcal{E})]_-<\infty$ as well, otherwise 
$EU(Y^n_T+\mathcal{E})\to -\infty$ would hold by \eqref{tomor}
along a subsequence, which would contradict the choice of $Y^n$. 
Hence \eqref{mam1} can be established
also for $U$ not bounded above. We then follow the proof of Theorem \ref{main1}.

Note that
the $\overline{Y}^n$ are convex combinations of the $Y^n$ so  
\[
\sup_n E_Q |\overline{Y}^n_T+\mathcal{E}|<\infty.
\]
We claim that the family 
\begin{equation}\label{uniff}
[U(\overline{Y}^n_T+\mathcal{E})]_+,\ n\geq 1,
\end{equation}
is uniformly integrable. Indeed, by \eqref{macci} and by \eqref{ulm},
\[
\sup_n E[\overline{Y}_T^n+\mathcal{E}]_+^{\theta}<\infty.
\]
Since $\theta>\alpha$, \eqref{maki} shows our claim.

It follows by the uniform integrability of \eqref{uniff}
and by \eqref{mable}
that $U(\overline{Y}^n_T+\mathcal{E})$ tends to $U(Y^{\ddagger}_T+\mathcal{E})$ in $L^1$ as $n\to\infty$
and we obtain the optimizer $Y^{\dagger}$ as before.
\end{proof}

Theorem \ref{main2} also has a version with $\mathcal{I}$ in lieu of $\mathcal{A}_U$.

\begin{theorem}\label{main2.5}
Let Assumptions \ref{u} and \ref{m2} be in force and let $\emptyset\neq\mathcal{I}\subset\mathcal{S}$
be convex and Fatou-closed. Define
$$
\mathcal{I}_U:=\{Y\in\mathcal{I}:\, EU(Y_T+\mathcal{E})>-\infty\}.
$$
Then there is $Y^{\dagger}\in \mathcal{I}_U$ such that
\[
EU(Y^{\dagger}_T+\mathcal{E})=\sup_{Y\in\mathcal{I}_U}EU(Y_T+\mathcal{E}),
\]
provided that $\mathcal{I}_U\neq\emptyset$.
\end{theorem}
\begin{proof} Note that $\mathcal{I}_U$ is convex.
We can follow the proof of Theorem \ref{main2} with $\mathcal{I}_U$ in lieu of $\mathcal{A}$, $\mathcal{A}_U$,
except for the end, where we use the uniform
integrability of \eqref{uniff} and Fatou's lemma to show that
$$
EU(Y^{\dagger}_T+\mathcal{E})\geq EU(Y^{\ddagger}_T+\mathcal{E})\geq\sup_{Y\in\mathcal{I}_U}EU(Y_T+\mathcal{E}).
$$
The other inequality being trivial (since $Y^{\dagger}\in \mathcal{I}_U$), the result follows.
\end{proof}

We may as well put constraints on the terminal portfolio wealth. This corresponds to e.g. regulations
imposed on the portfolio manager so we regard $\mathcal{K}$ in the next assumption as a set
of ``acceptable positions''. 

\begin{assumption}\label{k1}
The set $\mathcal{K}\subset L^0$ is convex and closed in probability.
\end{assumption}

\begin{example}\label{oppo} {\rm For instance, one can choose $\mathcal{K}:=\{X\in L^0:\ El(X_-)\leq K\}$ with some convex $l:\mathbb{R}_+\to
\mathbb{R}_+$ and $K>0$ or $\mathcal{K}:=\{X\in L^0:\ E[X-X^{\sharp}]^2\leq K\}$ with some fixed $X^{\sharp}\in L^0$, these satisfy
Assumption \ref{k1} by Fatou's lemma. The first example is a restriction on acceptable losses while the
second ensures that the investors' portfolio value is not far from a reference entity $X^{\sharp}$ (such as the
value of a benchmark portfolio). One may also define $\mathcal{K}:=\{X\in L^0:\ X\geq X^{\flat}\}$ with some
$X^{\flat}$ where $X^{\flat}$ provides an almost sure control on losses. Note that no integrability assumption
on $X^{\flat}$ is necessary.}
\end{example}

Define 
\[
\mathcal{S}':=\{Y\in\mathcal{S}:\ Y_T\in\mathcal{K}\}
\]
and let $\mathcal{A},\mathcal{I}\subset\mathcal{S}'$ be non-empty. Define 
\begin{eqnarray}\nonumber
\mathcal{A}_U' &:=& \{Y\in\mathcal{S}':\mbox{ there is }Y^n\in\mathcal{A}'\mbox{ with }U(Y^n_T+\mathcal{E})\in L^1,\ n\geq 1\\
& &\mbox{ and }U(Y^n_T+\mathcal{E})\to U(Y_T+\mathcal{E}),\, n\to\infty\mbox{ in }L^1\}.\nonumber 
\end{eqnarray}

\begin{theorem}\label{const} Under Assumption \ref{k1}, Theorems \ref{main1} and \ref{main2} hold when
$\mathcal{A}_U$ is replaced by 
$\mathcal{A}_U'$ provided that $\mathcal{A}_U'\neq\emptyset$.
\end{theorem}
\begin{proof}
We can verbatim follow the respective proofs noting 
that $Y^n$, $\tilde{Y}^n$
all stay in $\mathcal{S}'$, by Assumption \ref{k1}. Hence
the limit $Y^{\ddagger}$ is such that $Y^{\dagger}_T=Y^{\ddagger}_T\in\mathcal{K}$, again by
Assumption \ref{k1}.
\end{proof}

\begin{assumption}\label{k11}
The set $\mathcal{K}\subset L^0$ is convex and closed in probability, 
satisfying $\mathcal{K}+L_+^0\subset\mathcal{K}$.
\end{assumption}

\begin{remark}{\rm The sets 
$$
\{X\in L^0:\ El(X_-)\leq K\}\mbox{ and }
\{X\in L^0:\ X\geq X^{\flat}\}
$$
from Example \ref{oppo} both satisfy Assumption \ref{k11}.}
\end{remark}

\begin{theorem}\label{constt} Under Assumption \ref{k11},
Theorems \ref{main1.5} and \ref{main2.5} hold with $\mathcal{I}$ (resp. $\mathcal{I}_U$) replaced by 
$$
\mathcal{I}':=\{Y\in\mathcal{I}:\, Y_T\in\mathcal{K}\mbox{ a.s.}\}\ \mbox{(resp. }
\mathcal{I}_U':=\{Y\in\mathcal{I}_U:\, Y_T\in\mathcal{K}\mbox{ a.s.}\}\mbox{)},
$$
in their statements.
\end{theorem}
\begin{proof}
As in the previous proof, $Y^{\ddagger}_T\in\mathcal{K}$ a.s. hence 
$Y^{\dagger}_T\in \mathcal{K}+L_+^0\subset \mathcal{K}$ a.s.
\end{proof}

It seems problematic even to formulate the dual problem with general constraint set $\mathcal{K}$. 
Hence we doubt that Theorems \ref{const} and \ref{constt} could be shown by solving the dual problem
first and then returning to the primal problem. Our method, however, 
operates only
on the primal problem and it applies easily to the case with constraints as well.

\section{Frictionless markets}\label{fric}

Let $S_t$, $t\in [0,T]$ be an $\mathbb{R}^d$-valued semimartingale
on the given stochastic basis; $L(S)$ denotes the corresponding class of $S$-integrable processes.
When $H\in L(S)$, we use the notation $H\cdot S_u$, to denote the value of the stochastic integral of
$H$ with respect to $S$ on $[0,u]$, $0\leq u\leq T$. The process $S$ represents the price of $d$
risky securities, $H$ plays the role of an investment strategy and $H\cdot S_u$ is the value
of the corresponding portfolio at time $u$ (we assume that there is a riskless asset with
price constant one and that trading is self-financing).

We denote by $\mathcal{M}^a$ the set of $Q\ll P$ such that $S$ is a $Q$-local martingale.
Set $\mathcal{M}^e:=\{Q\in\mathcal{M}^a:\ Q\sim P\}$.
The process $S$ is not assumed to be locally bounded but, for reasons of simplicity, we
refrain from exploring the universe of sigma-martingales in this paper.
For this section, we make the choice
\[
\mathcal{M}^a_V:=\mathcal{M}^a\cap\mathcal{P}_V,
\]
see \eqref{mathp}. Set also $\mathcal{M}^e_V:=\mathcal{M}^e\cap\mathcal{P}_V$.
We recall an important closure property for stochastic integrals.

\begin{lemma}\label{clos} Let $Q\in\mathcal{M}^e$ and let $w_t\geq 1$, $t\in [0,T]$ be a $Q$-martingale. 
If $H^n\in L(S)$,
$n\geq 1$ is a sequence such that $H^n\cdot S_T\to X$ $P$-almost surely (which is the same as
$Q$-almost surely) for some $X\in L^0$ and 
\begin{equation}\label{matra}
H^n\cdot S_t\geq -w_t, 
\end{equation}
$P$-almost surely for all $n\geq 1$, $t\in [0,T]$ then there is $H\in L(S)$ and $N\in L^0_+$
such that $X=H\cdot S_T-N$.
\end{lemma}
\begin{proof}
When \eqref{matra} holds with a fixed $Q$-integrable random variable $w$ instead of $w_t$ then this result
is just a reformulation of Corollary 15.4.11 from \cite{ds}. One can check that the proof of that result
goes through with minor modifications under the conditions stated in the present lemma, too.
\end{proof}

\begin{remark}\label{burt} {\rm In the setting of frictionless markets, we now compare our Assumption \ref{m0} to 
those of \cite{btz} and \cite{oz}.
In \cite{btz} $U$ was not assumed either smooth or strictly concave but $\mathcal{E}$ had to be
bounded. For unbounded random endowments \cite{oz} seems to present the state-of-the-art as far
as the existence of optimal portfolio strategies is concerned.
That paper assumes continuous differentiability and strict concavity of $U$. On $\mathcal{E}$
they stipulate their Assumption 1.6 which reads as 
\begin{equation}\label{zi}
x'+H'\cdot S_T\leq \mathcal{E}\leq x''+H''\cdot S_T,
\end{equation}
with $x',x''\in\mathbb{R}$ and with $H',H''\in L(S)$ such that $H'\cdot S$ is a martingale and $H''\cdot S$ is a supermartingale, under each 
$R\in\mathcal{M}^a_V$.

Assumption \ref{m0} allows certain important cases of $\mathcal{E}$ which are excluded by \eqref{zi}:
we only require $E_R|\mathcal{E}|<\infty$ for all $R\in\mathcal{M}^a_V$ while \eqref{zi} implies 
$\sup_{R\in\mathcal{M}^a_V}E_R |\mathcal{E}|<\infty$, see Example \ref{unbi} for more on this.}
\end{remark} 

\begin{example}\label{unbi} {\rm Let the filtration be generated by two independent Brownian motions
$W_t$ and $B_t$, $t\in [0,T]$ and let the price of the single risky asset be given by $S_t:=W_t+t$
(we could take a drift other that $t$, we chose this one for simplicity).
Define the random endowment $\mathcal{E}:=B_T$. Define $U(x)=-x^2$ for $x\leq 0$ and $U(x)=0$, $x>0$.
Choose, for $n\geq 1$, $Q_n$ as the unique element of $\mathcal{M}^e$ such that
$B_t-nt$, $t\in [0,T]$ is a $Q_n$-Brownian motion and $Q_n\sim P$. It is easily checked that $Q_n\in\mathcal{M}^e_V$.
We trivially have \eqref{rom2}
but $E_{Q_n}\mathcal{E}=nT\to\infty$ as $n\to\infty$
so \eqref{zi} cannot hold. We can thus find optimizers using Theorem 
\ref{mai1} below even in cases where $\mathcal{E}$ constitutes a non-hedgeable risk in the sense that 
there are no $H'$, $H''$
satisfying \eqref{zi}. An analogous argument applies to a larger family of random endowments: e.g. the same can be concluded about 
$\mathcal{E}:=f(S_T)B_T$ for an arbitrary
bounded measurable $f$ such that $E_{Q_n}f(S_T)\neq 0$ (note that this expectation is
independent of $n$).}
\end{example}

Define
\[
\mathbb{S}:=\{H\in L(S):\ H\cdot S\in\mathcal{S}\},
\]
where $\mathcal{S}$ is as in \eqref{chante}.

\begin{theorem}\label{mai1}
Let Assumptions \ref{u} and \ref{m0} be in force and let $U$ be bounded above.
Then there exists $H^{\dagger}\in\mathbb{S}$ such that
\begin{equation}\label{passengers}
EU(H^{\dagger}\cdot S_T+\mathcal{E})=\sup_{H\in\mathbb{S}}EU(H\cdot S_T+\mathcal{E}).
\end{equation}
\end{theorem}

\begin{remark} {\rm Optimization problems like \eqref{passengers} arise in the study of indifference pricing and indifference hedging, see e.g. \cite{bfg}. 
We note that in \cite{oz} the domain of optimization was also $\mathbb{S}$. 

Notice that in \cite{oz} $S$ is assumed locally bounded while we do not need this hypothesis.
In \cite{oz} it was shown that  
\begin{equation}\label{bab}
\sup_{H\in\mathbb{S}} EU(H\cdot S_T+\mathcal{E})=\sup_{H\in\mathbb{A}^{\mathrm{adm}}}EU(H\cdot S_T+\mathcal{E}),
\end{equation}
where $\mathbb{A}^{\mathrm{adm}}$ is the set of portfolio strategies $H$ for which $H\cdot S$ is bounded
from below by a constant. In our setting, $S$ may fail to be locally bounded hence \eqref{bab} is clearly
false in general.}
\end{remark}

\begin{proof}[Proof of Theorem \ref{mai1}.] Set 
$$\mathcal{I}:=\{H\cdot S_T:\ H\in\mathbb{S}\},$$ 
this is convex and it is also Fatou-closed by Lemma \ref{clos}. Theorem \ref{main1.5} now implies the result.
\end{proof}

\begin{example}\label{lil} {\rm We call $\ell:\mathbb{R}_+\to\mathbb{R}_+$ a \emph{nice loss function} if it is a
Young function such that its conjugate $\ell^*$ is a Young function 
of class $\Delta_2$.  Typical nice loss functions are 
$\ell(x)=x^{\kappa}$ for some $\kappa>1$ 
(as their conjugate is also constant times a power function).

Minimizing the expected loss of a portfolio consists in finding $H^{\dagger}$ with
\begin{equation}\label{op}
E\ell([H^{\dagger}\cdot S_T +\mathcal{E}]_-)=\inf_{H\in\mathbb{S}}E\ell([H\cdot S_T +\mathcal{E}]_-).
\end{equation}
Theorem \ref{mai1}
applies here with the choice $U(x):=0$, $x>0$, $U(x):=-\ell(-x)$, $x\leq 0$, under Assumption \ref{m0}.

The existence of an optimal portfolio in general incomplete semimartingale models has already been considered
for such loss functions in the literature, see e.g. \cite{fl} and \cite{ph}. However, in these articles only portfolios with
a non-negative value process were admitted. Without this restriction, \cite{sara,sara_orlicz} cover the case
$\mathcal{E}=0$ and results of \cite{btz} apply when $\mathcal{E}$ is bounded. Our paper seems to be the first
to treat an unbounded random endowment in the context of loss minimization for value processes that are possibly not
bounded from below.}
\end{example}

\begin{theorem}\label{mai2} Let Assumptions \ref{u} and \ref{m2} be in force.
Define 
$$
\mathbb{S}_U:=\{H\in L(S):\ H\cdot S\in\mathcal{S},\ EU(H\cdot S_T+\mathcal{E})>-\infty\}.
$$
If $\mathbb{S}_U\neq\emptyset$ then
there exists $H^{\dagger}\in\mathbb{S}_U$ such that
\[
EU(H^{\dagger}\cdot S_T+\mathcal{E})=\sup_{H\in\mathbb{S}_U}EU(H\cdot S_T+\mathcal{E}).\quad\Box
\]
\end{theorem}
\begin{proof}
This follows from Theorem \ref{main2.5}. 
\end{proof}

We can also obtain the following result.

\begin{theorem}\label{conss} Under Assumption \ref{k11}, Theorems \ref{mai1} and \ref{mai2} hold when
$\mathbb{S}$ (resp. $\mathbb{S}_U$) is replaced by 
\[
\mathbb{S}':=\{H\in\mathbb{S}:\ H\cdot S_T\in\mathcal{K}\}\ \mbox{(resp. }
\mathbb{S}_U':=\{H\in\mathbb{S}_U:\ H\cdot S_T\in\mathcal{K}\}\mbox{)}, 
\]
in their statements.
\end{theorem}
\begin{proof}
This follows from Theorem \ref{constt}.\end{proof}

\begin{remark}\label{cerny}
{\rm In the extensive related literature, perhaps the approach of \cite{sara} is the closest to ours in spirit. 
In that paper the focus is on working with a pleasant class of admissible strategies while we stay within the 
``standard'' class of \cite{oz}. At the purely technical level, the main difference is that in \cite{sara}
the dual problem is formulated, the dual minimizer is found and then it plays an
important role in the construction of the primal optimizer. In our paper, thanks to
the results of \cite{delbaen-owari}, we avoid introducing the dual problem altogether.
This is advantageous since, quite often, the dual problem is difficult to analyse (as in the
case of random endowments when the space of finitely additive measures needs to be used, 
see \cite{oz}) or even hopeless to properly formulate (as in the case of constraints, see 
Theorem \ref{conss} above).}
\end{remark}

\section{Large markets}\label{lar}

The methods of the present paper are also applicable to markets with frictions, even
in the presence of model ambiguity, see \cite{robust}. Here we present another application, to models with infinitely many assets.

Large financial markets were introduced in \cite{kabanov-kramkov} as a sequence of market models with a
finitely many assets. For a review of the related literature
we refer to \cite{josef,callum}. In the present paper we only treat the case where all the countably many
assets are defined on the same probability space. 

Staying in the setting of Section \ref{ut},
let $S^j_t$, $j\geq 1$, $t\in [0,T]$ be a sequence of $\mathbb{R}$-valued semimartingales on the given stochastic basis
$(\Omega,\mathcal{F},(\mathcal{F}_t)_{t\in [0,T]},P)$. We denote by $\mathcal{M}^a$ the set of $Q\ll P$
such that $S^j$ is a $Q$-martingale for each $j\geq 1$. Let $\mathcal{M}^e:=\{Q\sim P:\ Q\in\mathcal{M}^a\}$,
$\mathcal{M}^a_V:=\mathcal{M}^a\cap\mathcal{P}_V$ and $\mathcal{M}^e_V:=\mathcal{M}^e\cap\mathcal{P}_V$.

\begin{remark} {\rm It is shown in \cite{josef} that $\mathcal{M}^e\neq\emptyset$ can be characterized 
by the absence of free lunches with vanishing risk. Hence $\mathcal{M}_V^e\neq\emptyset$ is a
``strengthened'' no-arbitrage assumption, taking into account the given investor's preferences via $V$,
the conjugate of the utility function $U$.}
\end{remark}

Define the $\mathbb{R}^m$-valued semimartingale
$F^m_t:=(S^1_t,\ldots,S^m_t)$ and set
\[
\mathbb{A}^m:=\{H\in L(F^m):\ H\cdot F^m_t\geq -s\mbox{ for all }t\in [0,T]\mbox{ with some }s>0\},\ m\geq 1,
\] 
where $L(F^m)$ denotes the set of $F^m$-integrable processes. It is implicitly assumed that
there is a riskless asset of price constant $1$ and that trading is self-financing, hence $H\cdot F^m$
is the value process of a portfolio in the risky assets $S^1,\ldots,S^m$ corresponding to the strategy $H$,
starting from zero initial capital.

It is natural to take 
$$
\mathcal{A}:=\cup_{m\geq 1}\{H\cdot F^m:\ H\in\mathbb{A}^m\}
$$
but $\mathcal{A}$ can't serve as a domain of optimization since
$\{Y_T:\, Y\in\mathcal{A}\}$ is not closed in any reasonable topology.
Following the papers \cite{paolo,mostovyi}, we resort to generalized strategies. The novelty
is that \cite{paolo,mostovyi} consider utilities on the positive real axis while we are
able to treat utilities $U:\mathbb{R}\to\mathbb{R}$, for the first time in the related
literature.

Recall the definition of $\mathcal{S}$ from \eqref{chante} and note that $\mathcal{A}\subset\mathcal{S}$ by \cite{as}. Let us recall the definition of $\mathcal{A}_U$ from Section \ref{ut}:
\begin{eqnarray*}
\mathcal{A}_U  &:=& \{Y\in\mathcal{S}:\mbox{ there is }Y^n\in\mathcal{A}\mbox{ with }U(Y^n_T+\mathcal{E})\in L^1,\ n\geq 1\\
& & \mbox{ and }U(Y^n_T+\mathcal{E})\to U(Y_T+\mathcal{E})\mbox{ in }L^1\}.\nonumber 
\end{eqnarray*}

Identifying portfolios with their value processes, we call elements of $\mathcal{A}_U$ \emph{generalized portfolio strategies}. With this choice of $\mathcal{A}$,  
Theorems \ref{main1}, \ref{main2} and \ref{const} prove the existence of optimizers in the
class of generalized strategies for a large financial market. 

\begin{remark}\label{unco}
{\rm In the present setting, it is crucial from the point of view of economic interpretations 
that the optimizer can be approximated by portfolios in finitely many assets, i.e. 
the optimizer lies in $\mathcal{A}_U$. That's why we apply Theorems \ref{main1}, \ref{main2} and \ref{const} and not
Theorems \ref{main1.5}, \ref{main2.5} or \ref{constt} 
where the class $\mathcal{I}$, a
priori, does not have any feature of 
``approximability by admissible strategies''. The price we pay is that Assumption \ref{moder} needs
to be posited which is a restriction on the tail of $U$ at $-\infty$.

Markets with uncountably many assets can also be treated in an analogous
manner, as easily seen. We confined ourselves to the countable case only, 
to stress connections with the extensively studied
area of large financial markets. 

Instead of $\mathbb{A}^m$, one could take portfolios whose value
processes are bounded
from below by constant times a weight function. This is a reasonable choice for 
price processes that are not locally bounded, see e.g. \cite{bifri} or Chapter 14 of \cite{ds}.}
\end{remark}

\end{document}